\documentclass[]{llncs}
\usepackage{float}
\floatstyle{boxed}
\newfloat{figure}{h}{ext}
\floatname{figure}{Figure}
\usepackage{graphicx} 
\usepackage{amssymb}
\usepackage{pstricks}
\usepackage{multirow}

\newcommand{\lang}{{\cal L}}

\newcommand{\dlvrd}{{\tt dlvrd}}
\newcommand{\rr}{{\tt rr}}
\newcommand{\kc}{{\tt kc}}

\newcommand{\sender}{{\tt sender}}

\newcommand{\rcvd}{{\tt rcvd}}

\newcommand{\conflictfree}{{\tt conflict\mbox{-}free}}

\newcommand{\rimp}{\Rightarrow}
\newcommand{\dimp}{\Leftrightarrow}

\newcommand{\slotrequest}{{\tt slot\mbox{-}request}}
\newcommand{\msg}{{\tt message}}
\newcommand{\conflict}{\tt conflict}

\newcommand{\be}{\begin{enumerate}}
\newcommand{\ee}{\end{enumerate}}

\newcommand{\xor}{\otimes}
\newcommand{\keys}{{\it keys}}
\newcommand{\rand}{{\it rand}}
\newcommand{\broadcast}{{\it broadcast}}
\newcommand{\omap}{{\it ov}}
\newcommand{\Agt}{{\it Agt}}
\newcommand{\Var}{{\it Var}}

\newcommand{\xm}{\xor m}

\begin{document}

\title{Abstraction for Epistemic Model Checking of Dining Cryptographers-based 
Protocols
\thanks{ This material is based on research sponsored by the Air Force Research Laboratory, under agreement number FA2386-09-1-4156. The U.S. Government is authorized to reproduce and distribute reprints for Governmental purposes notwithstanding any copyright notation thereon. The views and conclusions contained herein are those of the authors and should not be interpreted as necessarily representing the official policies or endorsements, either expressed or implied, of the Air Force Research Laboratory or the U.S. Government.
\today}}

\author{Omar I. Al-Bataineh and Ron van der Meyden}
\institute{School of Computer Science and Engineering,\\
University of New South Wales}
\maketitle

\begin{abstract}
The paper describes an abstraction for protocols that are based on multiple rounds of 
Chaum's Dining Cryptographers protocol. It is proved that the abstraction preserves a rich class of 
specifications in the logic of knowledge, including specifications 
describing what an agent knows about other agents' knowledge. 
This result can be used to optimize model checking of 
Dining Cryptographers-based protocols, and applied within a methodology 
for knowledge-based program implementation and verification. 
Some case studies of such an application are given,  
for a protocol that uses the Dining Cryptographers protocol 
as a primitive in an anonymous broadcast system. 
Performance results are given for model checking 
knowledge-based specifications in the concrete and abstract models of this protocol, and some
new conclusions about the protocol are derived. 
\end{abstract}

\section{Introduction}

Relations of abstraction (and their converse, refinement) 
are valuable tools for program verification. In this approach, 
we relate a (structurally complex) concrete program to a (simpler) abstract program
by means of a relation that is known to preserve the properties that we wish to verify 
in the concrete program. When such a relation can be shown to hold, 
we are able to verify these properties in the concrete program by 
showing that they hold in the abstract program, which is 
generally easier in view of the lesser structural complexity of the 
abstract program. In particular, model checkers can be expected to 
run more efficiently on the abstract program than on the concrete program, 
and abstraction is often used to bring the verification problem within the 
bounds of feasibility  for model checking. 
Conversely, starting with the abstract program, and 
having verified that this satisfies the desired properties, 
we may derive the concrete program and conclude that this also satisfies 
these properties. This perspective is the basis for ``correctness-by-construction" or top-down refinement 
approaches to program verification.  

Our contribution in this paper is to establish the correctness of an abstraction relation 
for abstract programs based on use a trusted third party for anonymous broadcast, 
which is implemented in the related concrete programs by means of the {\em Dining Cryptographers} protocol 
proposed by Chaum \cite{chaum}. That Chaum's  protocol implements anonymous 
broadcast is, of course, well-known, but we show that this statement holds in a more
general sense than is usually considered in the literature, where the focus is generally on the very particular property of anonymity. 
Specifically, we consider a broad class of properties formulated in the logic of knowledge, including properties in which agent knowledge is nested, 
such as ``Alice knows that Bob knows that $p$". We show that the abstraction relation 
between programs based on the trusted third party and programs based on 
the Dining Cryptographers protocol preserves all properties from this class. 

As an application of this result, we consider a protocol from Chaum's paper \cite{chaum} 
that uses multiple rounds of the Dining Cryptographers protocol to build a more general 
anonymous broadcast system. We have previously studied this protocol from the 
perspective of  a model checking based methodology for the implementation of 
knowledge-based programs \cite{AlBatainehMeyden10}, by treating the specification of the 
protocol as a knowledge-based program containing nested knowledge formulas. 

Knowledge-based programs \cite{FHMVbook} are an abstract, program-like form of specification, 
that describe how an agent's actions are related to conditions stated in terms of the agent's knowledge. The advantage of this 
level of abstraction is that it provides a highly intuitive description of the intentions 
of the programmer, that has been argued to be easier to verify than the complex 
implementations one typically finds for highly optimized distributed programs \cite{HZ92,FHMVbook}. 
Knowledge-based programs cannot be directly implemented, however, so they 
must be {\em implemented} by concrete programs in which the 
knowledge  conditions are replaced by concrete predicates of the agent's local state. 
The implementation relation between a knowledge-based program and a putative implementation 
holds when these concrete predicates are equivalent to the knowledge formulas that 
they replace (interpreted with respect to the system generated by running the putative implementation). 
Our partially-automated methodology for the implementation of knowledge-based programs 
uses a model checker for the logic of knowledge  to check whether this equivalence holds, 
and if it does not, uses the counter-examples generated by the model checker to 
generate a revised putative implementation. (This process is iterated until an implementation is found.) 

In our previous work on the application of this methodology, we considered 
model checking problems generated in this way from a knowledge-based 
program based on multiple rounds of the Dining Cryptographers protocol. 
Our experience was that the model checking problems we considered  
were close to the bounds of feasibility for our model checker even for instances with 
small numbers of agents, and we were prevented from considering instances of scale as a result. 
In the present paper, we apply our abstraction result in order to optimize the 
model checking problem, by performing model checking on the abstracted 
(trusted third party) version of the programs we consider rather than the concrete 
(Dining Cryptographers based) versions.  We give performance results
showing the difference, which indicate that the abstraction is effective in 
reducing the model checking runtime by several  orders of magnitude, enabling 
systems involving larger numbers of rounds of the  Dining Cryptographers protocol
and larger numbers of agents to be model checked. We use the efficiency gains to extend 
our previous analysis of the knowledge based program to larger numbers of agents, 
leading to an improved understanding of its implementations. 

The structure of the paper is as follows. We begin in Section~\ref{sec:knowledge}  
by introducing the logic of knowledge, which provides the specification language for the
properties that are preserved by our abstraction technique, and give its semantics in 
terms of a class of Kripke structures. We define a notion of bisimulation
on these Kripke structures that provides the semantic basis for our program abstraction technique. 
In Section~\ref{sec:programs}, we introduce a simple programming language 
used to represent our concrete and abstract programs.  In Section~\ref{sec:DC}, we 
introduce the Dining Cryptographers protocol and, in Section~\ref{sec:DCa},   its abstraction using a trusted 
third party. In Section~\ref{sec:abstraction} we state and prove correct the abstraction relation.  
The remainder of the paper deals with our application of this result. 
We recall the two-phase protocol in Section~\ref{sec:rounds}. In Section~\ref{sec:kbp} we describe knowledge-based programs 
and an approach to the use of model checking to identify their implementations. 
In Section~\ref{sec:dckbp} we recall our formulation of the two-phase protocol 
as a knowledge-based program and describe the associated verification conditions. 
Section~\ref{sec:results} discusses the comparative performance of
model checking in the concrete and abstract models when using the model checker MCK. 
We highlight some of the interesting conclusions we are able to make about implementations 
of the knowledge-based program for the round-based protocol in Section~\ref{sec:imp-results}. 
We discuss related work in Section~\ref{sec:related}. Finally, in Section~\ref{sec:concl},  we draw some conclusions and 
discuss future directions.  

\section{Epistemic Logic and Bisimulations}\label{sec:knowledge} 

Epistemic logics are a class of modal logics that
include operators whose meaning concerns
the information available to agents in a 
distributed or multi-agent system. 
In epistemic model checking, one is generally concerned
with the combination of such operators with 
temporal operators, and a semantics using 
a class of structures  known in the literature as  {\em interpreted 
systems} \cite{FHMVbook} that combines temporal and epistemic 
expressiveness. We focus here on a simpler framework that omits temporal operators, 
since we are mostly interested, in our application, on what knowledge agents have after some program has run, 
and this also simplifies the statement and proof of our results. 

Suppose that we are interested in systems comprised of  agents
from a set $\Agt$ whose states are described using a set $\Var$ of 
boolean variables.\footnote{We use the term ``variable" rather than ``proposition" in this paper, since our
atomic propositions arise as boolean variables in a program.} 
The syntax of the  logic of knowledge $\lang_{(\Var, \Agt)}$
is given by the following grammar: 
\[
\phi ::= \top \mid v \mid \neg \phi \mid \phi \wedge \phi \mid K_i\phi 
\]
where $v\in \Var$ is a variable and $i\in \Agt$ is an agent. 
(We freely use standard boolean operators that can be defined 
using the two given.) Intuitively, the meaning of $K_i\phi$ 
is that agent $i$ knows that $\phi$ is true. 

The semantics for the language is given in terms of {\em Kripke structures} 
of the form  
$M = (\Agt, W, \{\sim_i\}_{i\in \Agt}, \Var, \pi)$, 
where 
\be 
\item $\Agt$ is the set of agents, 
\item $W$ is a set of worlds, or situations,  
\item for each $i\in \Agt$, $\sim_i$ is an equivalence relation on $W$,
\item $\Var$ is a set of variables, 
\item $\pi: W\times\Var \rightarrow \{0,1\}$ is a valuation. 
\ee 
Intuitively, $W$ is the set of situations that the agents
consider that they could be in, and $w\sim_iw'$ if, 
when the actual situation is $w$, agent $i$ considers it possible that 
they are in situation $w'$. The value $\pi(w,v)$ is the truth value of variable
$v$ in situation $w$. Such a Kripke structure $M$ is {\em fit} for the 
language $\lang_{(\Var',\Agt')}$ if $\Agt'\subseteq \Agt$ and 
$\Var'\subseteq \Var$.  The semantics of the language is given by the 
relation $M,w\models \phi$, where $M$ is a Kripke structure fit for $\lang_{(\Var,\Agt)}$, $w$ 
is a world of $M$, 
and $\phi$ is a formula, meaning intuitively that the formula $\phi$ holds at the 
world $w$. The definition is given inductively by 
\be 
\item $M,w\models v$ if $\pi(w,v)=1$, for $v\in \Var$. 
\item $M,w\models \neg \phi$ if not $M,w\models \phi$, 
\item $M,w\models \phi_1\land \phi_2$ if $M,w\models \phi_1$ and $M,w\models \phi_2$, 
\item $M,w\models K_i\phi$ if $M,w'\models \phi$ for all $w'\in W$ with $w\sim_i w'$, for $i\in \Agt$. 
\ee
Intuitively, the final clause says that 
agent $i$ knows $\phi$ if it does not consider it possible that not $\phi$. 
We write $M\models \phi$, and say that $\phi$ is {\em valid} in $M$,   if $M,w\models \phi$ for
all $w\in W$. The {\em Kripke structure model checking} problem  is to compute, given $M$ and $\phi$, 
whether $M\models \phi$.  We will use this formulation of the model checking problem as the basis for another 
notion of model checking, to be introduced below, that concerns a way of generating $M$ from a program. 

One of the difficulties to be faced in model checking, the {\em state space explosion} problem, 
is the potentially large size of the set of worlds $W$ of the structures $M$ of interest. 
Abstractions are useful techniques for mitigating
state space explosion problem. They are often applied 
as a preliminary step to model checking. 
Systems often encode details
that are irrelevant to the properties that we aim
to verify. Abstraction techniques enable us to eliminate such unnecessary,  redundant details. 
However, abstractions must be sound, in the sense that properties that hold in the abstract model must
also hold in the concrete model.

For Kripke structures, {\em bisimulations}
may provide an 
effective way to simplify redundant structure 
while preserving properties of interest. We formulate here a 
version that is suited to our application, in which we allow both the 
set of agents and the set of propositions to vary  in the structures 
we consider. 

Suppose we are given a set of variables $\Var$, a set of agents $\Agt$, 
and two Kripke structures 
$$M = (\Agt^M, W^M, \{\sim^M_i\}_{i\in \Agt^M}, \Var^M, \pi^M)$$  
and  $$N = (\Agt^N, W^N, \{\sim^N_i\}_{i\in \Agt^N}, \Var^N, \pi^N)$$ 
such that $\Agt \subseteq \Agt_M\cap  \Agt_N$ and $\Var \subseteq \Var^M \cap \Var^N$. 
(Note that these conditions imply that both $M$ and $N$ are fit for $\lang_{(\Var,\Agt)}$.) 
A $(\Var,\Agt)$-bisimulation $\Re$ between $M$ and $N$ is defined to be a binary relation
$\Re \subseteq W^M \times W^{N}$ such that:
\begin{enumerate}
\item \textbf{Atoms:} $\pi^{M} (w, v) = \pi^{N} (w', v)$ whenever $w \Re w'$ and $v\in \Var$; 

\item  \textbf{Forth:}  if $i\in \Agt$,  and $w_{1}, w_{2}$ are two worlds in $M$ and $u_{1}$ is a world in $N$
such that $w_{1} \sim^M_{i} w_{2}$ and $w_{1} \Re u_{1}$, then there is
a world  $u_{2} \in W_{N}$ such that $u_{1} \sim^N_{i} u_{2}$
and $ w_{2} \Re u_{2}$; and 
\item \textbf{Back:}  if  $i\in \Agt$ and $u_{1}, u_{2}$ are two worlds in $N$ and $w_{1}$ is a world in $M$
such that $u_{1} \sim^N_{i} u_{2}$ and $u_{1} \Re w_{1}$, then there is
a $w_{2} \in W_{M}$ such that $w_{1} \sim^M_{i} w_{2}$
and $ u_{2} \Re w_{2}$.

\end{enumerate}
If there exists an $(\Var,\Agt)$-bisimulation $\Re$ between $M$ and $N$ such that $w \Re u$,
then we write $(M,w) \approx_{(\Var,\Agt)} (N, u)$. 
If there  exists an $(\Var,\Agt)$-bisimulation $\Re$ between $M$ and $N$ such that 
for every $u\in W^M$ there exists $w\in W^N$ such that $u\Re w$ and, conversely, 
for every $w\in W^N$ there exists $u\in W^M$ such that $u\Re w$,  then we write $M\approx_{(\Var,\Agt)} N$. 
The following result shows that $(\Var,\Agt)$-bisimulation preserves properties in the language $\lang_{(\Var,\Agt)}$. 

\begin{lemma}
If $M$ and $N$ are Kripke structures and $u$ and $w$ are worlds of 
$M$ and $N$ such that
 $(M,u) \approx_{(\Var,\Agt)} (N, w)$, then 
 for all $\varphi \in \lang_{(\Var,\Agt)}$ we have 
$M, u\models \varphi$ if and only if $N, w \models \varphi$. 
If $M\approx_{(\Var,\Agt)} N$ then for all $\varphi \in \lang_{(\Var,\Agt)}$ we have 
$M \models \varphi$ if and only if $N \models \varphi$.
\end{lemma}
 
We omit the proof since it is a minor variant of well-known results 
in the literature.  
In our applications of this result, 
we will consider a complex, concrete structure $M$ 
and a simper, more abstract structure $N$, and show that 
$M\approx_{(\Var, \Agt)} N$. This enables us to verify 
$M \models \varphi$ using the model checking problem
$N \models \varphi$, which is likely to be computationally 
easier in view of the smaller size of $N$. However, we need 
to also develop an abstraction technique for the programs that 
generate these Kripke structures. We develop this technique in the following sections. 

\section{A Programming Language and its Semantics} \label{sec:programs} 

We use a small multi-agent programming language equipped with a notion of observability. 
All variables are Boolean,  and expressions are formed from variables  using the usual Boolean operators. 
The language has the following 
atomic actions, 
in which $i$ and $j$ are  agents,  $x$ is a variable name and $e$ is an expression: 
\be 
\item $i: x:=e$  ---  agent $i$ evaluates $e$ and assigns the result to $x$, 
\item $i: \rand(x)$ --- agent $i$ assigns a random value to $x$,  
\item $i: e \rightarrow j.x$ --- agent $i$ evaluates $e$ and transmits the result across a private channel  to agent $j$, 
 who assigns it to its variable $x$, 
 \item $i: \broadcast(x)$ --  agent $i$ broadcasts the value of the variable $x$ to all other agents. 
\ee
Note that we write $i.x$ for agent $i$'s variable $x$ (the variables $i.x$ and $j.x$ are considered distinct when $i\neq j$)
but may omit the agent name when this is clear from the context.   In particular,  in
an atomic action $i:a$, any variable $x$ not explicitly associated with an agent
refers to $i.x$. For example, we may write $i: x := y \xor z$ rather than $i: i.x := i.y \xor i.z$. 
Similarly, when $e$ is an expression in which agent indices are omitted, and $i$ is an agent, 
the expression $i.e$ refers to the result of replacing each occurrence of a variable name $x$ in $e$ 
that is not already associated to an agent index 
with $i.x$. Thus $i.(y \xor j.z)$ represents  $i.y \xor j.z$.

 Each atomic action {\em reads}  and {\em writes}  certain variables. 
 Specifically, the action $i: x:=e$  reads the (agent $i$) variables in $e$ and writes $i.x$, the action 
$i: \rand(x)$ reads nothing and writes $i.x$, the action 
$i: e \rightarrow j.x$ reads the (agent $i$) variables in $e$ and writes $j.x$, and the action 
$i: \broadcast(x)$ reads $x$  and writes nothing. 
A {\em  joint action} is  a set of atomic actions in which no variable is 
written more than once.   Intuitively,  a joint action is executed by first evaluating all the expressions and then performing a  simultaneous
assignment  to the variables. 

A {\em program}  is given by a sequence of joint actions
$A_1; \ldots; A_n$. 
A {\em program for agent $i$} is a program in which  each atomic action $j:a$ in any step 
has $j=i$. We permit parallelism within an agent, in the sense that we do {\em not} require that 
a joint action contains at most one atomic action for each agent. 
If we are given for each agent $i$ a program $P_i= A^i_1; \ldots; A^i_n$, all of the same
length $n$, then we may form the joint program $||_i P_i = (\cup_i A^i_1); \ldots ; (\cup_i A^i_n)$. 

Some well-formedness conditions are required on agent programs. 
An {\em observability mapping} is a function $\omap$  mapping each agent 
to a set of variables, intuitively, the set of variables that it may observe. 
A program runs in the context of an observability mapping, and modifies that 
mapping. We say that {\em a joint action $A$ is  enabled} at an observability 
map $\omap$ if 
\be 
\item no variable written to by $A$ is in  $\omap(i)$ for any agent $i$
(that is, all variables written to are 
{\em new} variables), and 
\item for each atomic action $i: x:=e$ and $i: e \rightarrow j.x$ in $A$, 
the expression $i.e$ contains only variables in $\omap(i)$, and 
\item for each action $i.\broadcast(x)$ we have $i.x\in \omap(i)$. 
\ee 
These constraints may be understood as access control constraints stating that 
agent $i$ may read only the variables in $\omap(i)$ and may write only new variables. 

Executing the action $A$ transforms the  observability map $\omap$
to the  observability map $\omap [A]$  
such that $\omap[A](i)$ is  the result of adding to $\omap(i)$
\be 
\item 
all variables $i.x$  such that an action of the form $i: x:= e$ or $i:\rand(x)$ or  $j: e\rightarrow i.x$ occurs in $A$, and 
\item 
all variables $j.x$   such that $j:\broadcast(x)$ 
occurs in $A$. 
\ee 
These definitions are  generalised to programs: 
{\em the program $P = A_1; \ldots; A_n$  is enabled}  at
the observability map $\omap$  if for each $i = 1 \ldots n$, 
the action $A_i$ is  enabled at $\omap [A_1]\ldots [A_{i-1}]$, 
and we define $\omap [P]$ to be $\omap [A_1]\ldots [A_n]$.

\begin{example} 
Consider a two-agent system with agents $i,j$. 
The action $\{i: x:= j.y\}$ is not enabled at the observability map $\omap$ 
given by $\{j \mapsto \{j.y\}\}$. However, the 
program $\{j:\broadcast(y)\}; \{i: x:= j.y\}$ is enabled at $\omap$, 
since  the action $\{j:\broadcast(y)\}$ is enabled at $\omap$, and 
transforms $\omap$ to  $\omap[\{j:\broadcast(y)\}] = \{j \mapsto \{j.y\}, ~i\mapsto  \{j.y\}\}$,
at  which the action $\{i: x:= j.y\}$ is enabled. 
\end{example} 

We say that an  observability map is {\em  consistent} with a Kripke  structure
$M = (\Agt, W, \{\sim_i\}_{i\in \Agt}, \Var, \pi)$
when for all agents $i$, if $v$ is a variable in $\omap(i)$ then $v\in Var$, and for 
 all worlds $w,w'\in W$ such that $w\sim_i w'$
we have $\pi(w,v) = \pi(w',v)$. Intuitively, $\omap$ is consistent with $M$  if all 
variables declared to be local to agent $i$ by $\omap$ are in fact defined and  semantically 
local  to agent $i$ in $M$.

The {\em program $P$ is enabled at a Kripke structure $M$} if 
there exists an observability map $\omap$ such that
\be 
\item $\omap$  is consistent with $M$, 
\item $P$ is enabled at $\omap$, and 
\item all variables  $x$ written by $P$ are not defined in $M$ (i.e., $x\not \in Var$). 
\ee 
In particular, note that if a single joint action $A$ is enabled at $M$, then 
for all variables $x$ read by $A$, and all worlds $w$, the value $\pi(w,x)$ is
defined. Consequently, we may also evaluate at $w$ any expression $e$ 
required to be computed  by $A$. We write $\pi(w,e)$ for the result. 

We can now give a semantics of programs,  in which a program applied to a Kripke structure
representing the  initial states of information of the agents, 
transforms the structure into another Kripke structure  representing 
the states of information of the agents after running the program. The definition is 
given inductively, on an action-by-action basis. 
Let $M = (\Agt, W, \{\sim_i\}_{i\in \Agt}, \Var, \pi)$ be  a Kripke structure and 
$A$ a joint action. We define a Kripke structure 
$M [A] = (\Agt', W', \{\sim'_i\}_{i\in \Agt'}, \Var', \pi')$ as follows. 
Let $V$ be the set of variables $i.x$ such that $A$ includes
the atomic action $i:\rand(x)$. Intuitively, such actions increase the 
amount of non-determinism in the system, whereas all other actions  
have deterministic effects. We define
$\Agt' = \Agt$ and take  $W'$ to be the set of states of the form $(w, \kappa)$ where 
$w\in W$ and  $\kappa : V \rightarrow \{0,1\}$ is 
an assignment of boolean values to the variables  in $V$. 
We may write $w+\kappa$ for the pair $(w,\kappa)$. 
In case $V$ is the empty set, $\kappa$ is always the null
 function, so we 
may write just $w$ for $(w,\kappa)$. The set $\Var'$ of variables defined in 
$M[A]$ is obtained by adding to $\Var$ all variables written to by $A$. 
The assignment $\pi'$ is obtained by extending $\pi$ to these new variables
by defining $\pi'$ as follows on worlds $w+\kappa$: 
\be
\item if $v \in Var$ then $\pi'(w+\kappa,v ) = \pi(w,v)$ , 
\item if $i: x:=e$ occurs in $A$ then $\pi'(w+\kappa,i.x) = \pi(w,i.e)$ , 
\item if $i: \rand(x)$ occurs in $A$ then $\pi'(w+\kappa,i.x) = \kappa(i.x)$, and 
\item if $j: e \rightarrow i.x$ occurs in $A$ then $\pi'(w+\kappa,i.x) = \pi(w,j.e)$.  
\ee 
Finally, the indistinguishability relations $\sim'_i$ are defined using the observability map 
$\omap [A]$: we define 
$w+\kappa \sim'_i w'+\kappa'$ when $w\sim_i w$ 
and for all variables $x$ in $\omap[A](i)\setminus \omap(i)$, we have $\pi'(w+\kappa, x) = \pi'(w'+\kappa', x)$. 
Intuitively, this reflects that the agent recalls any information it had in the structure $M$, 
and adds to this information that it is able to observe in the new state. 
Note that in fact $w+\kappa \sim'_i w'+\kappa'$ implies $\pi'(w+\kappa, x) = \pi'(w'+\kappa', x)$
for {\em all} variables $x\in \omap[A](i)$, 
since we have assumed that for $x\in \omap(i)$ we have that  $w\sim_i w$ implies $\pi(w,x) = \pi(w',x)$. 
Moreover, since the set $\omap[A](i) \setminus \omap(i)$ is just the set of 
variables written to by $A$ that are made observable to $i$, this observation 
also yields that the definition of $M[A]$ is independent of the 
choice of observation map $\omap$ consistent with $M$. 

\section{Chaum's Dining Cryptographers Protocol}\label{sec:DC} 

Chaum's Dining Cryptographers protocol is an example of an anonymous broadcast
protocol: it allows an agent to send a message without revealing
its  identity. Chaum introduces the protocol 
with the following story:

\begin{quote}
Three cryptographers are sitting down to dinner at their favourite restaurant.
Their waiter informs them that arrangements have been made with the maitre d'hotel
for the bill to be paid anonymously. One of the cryptographers might be paying for
the dinner, or it might have been NSA (U.S National Security Agency). The three
cryptographers respect each other's right to make an anonymous payment, but they
wonder if NSA is paying. They resolve their uncertainty fairly by carrying out the following protocol:\\
Each cryptographer flips an unbiased coin behind his menu, between him and the cryptographer on his right, so that only the two of them can see the outcome. Each cryptographer then states aloud whether the two coins he can see--the one he flipped and the one his left-hand neighbor flipped--fell on the same side or on different sides. If one of the cryptographers is the payer, he states the opposite of what he sees. An odd number of differences uttered at the table indicates that a cryptographer is paying; an even number indicates that NSA is paying (assuming that the dinner was paid for only once). Yet if a cryptographer is paying, neither of the other two learns anything from the utterances about which cryptographer it is.
\end{quote}

\newcommand{\inedges}{{\it in}} 
\newcommand{\outedges}{{\it out}} 

Chaum shows that this protocol solves the problem, and notes that it can be considered
as a mechanism enabling a signal to be anonymously transmitted, under the  
assumption that at most one of the agents wishes to transmit. 
He goes on to generalize the idea to $n$-agent settings where, in 
place of the ring of coins, we have a graph 
representing the key-sharing arrangement. 

The more general protocol can be represented in our
programming language as follows. We assume that there is a set 
$\Agt$ of agents, who share secrets based on a (directed) key sharing graph 
$G = (\Agt, E)$ in which the vertices are the agents in $\Agt$ and the edges $E\subseteq \Agt \times \Agt$ 
describe the keysharing arrrangement amongst the agents. 
We model keysharing by assuming that for each edge $e = (i,j)$, 
agent $i$ generates the key corresponding to the edge, and 
communicates the key to $j$ across a secure channel. 
For each edge $e= (i,j)$ we write $e_1$ for the source agent $i$ and $e_2$ for the destination agent $j$. 
For each agent $i$ we define $\inedges(i) = \{e\in E ~|~e_2 = i\}$ and 
$\outedges(i) = \{e\in E~|~e_1 = i\}$.  Accordingly, we use two variables for each edge $e= (i,j)$: the variable 
$i.k_e$ stores $i$'s copy of the key corresponding to the edge,  and the variable 
$j.k_e$ stores  $j$'s copy. We write $\keys(i)$ for $\inedges(i) \cup \outedges(i)$, i.e., the set of edges incident on $i$. 
The protocol $DC_i(m)$ of an agent $i \in \Agt$ (in which the message represented by the expression $i.m$ 
is transmitted anonymously by agent $i$) consists of  the following five steps: 

\begin{figure} [h]
$$
DC_i(m)  = \begin{array}[t]{l} 
\{ i:\rand(k_e) ~|~ e\in \outedges(i) \}; \\ 
\{ i: k_e \rightarrow e_2.k_e ~|~ e\in \outedges(i) \} \\ 
\{ i: b : =  m \xor  \xor_{e\in \keys(i)}~ k_e \} ; \\ 
\{ i: \broadcast (b) \} ;   \\
\{ i: rr :=   \xor_{j\in \Agt} ~j.b\} 
\end{array} 
$$ 
\textbf{ \caption{The protocol $DC$}}
\end{figure}

We write $DC(m)$ for the joint program $||_{i\in \Agt} DC_i(m)$.

Intuitively, the protocol DC operates by first generating keys and setting up the key sharing graph, 
and then having each of the agents make a public announcement encrypted using all the keys 
available to them. The directionality of an edge in the key sharing graph indicates who
generates the key corresponding to the edge, viz, the source agent of the edge. 
The first step of the protocol corresponds to each agent generating the key values for which they are responsible.  
In the second step, these keys are shared with the other agent on the edge by transmission across a secure channel. 
Each agent now has the value of each of the key edges on which it is incident, and computes the 
xor of its message with all these key values in the 3rd step, and broadcasts the result in the 4th step. 
In the final step of the protocol, each agent computes the xor of the messages broadcast as the 
result of the protocol. 

\section{An Abstraction of the Dining Cryptographers Protocol} \label{sec:DCa} 

We are interested in protocols in which the DC protocol is used as a basic building block, 
and in model checking the agent's knowledge in the resulting protocols. 
In order to optimize this model checking problem, we now introduce a protocol 
that we will show to be an abstraction of the DC protocol that preserves epistemic properties. 

The abstracted version of the protocol omits the use of keys, but adds to the set of agents  a trusted third party
$T$ who computes the result of the protocol on behalf of the agents, and then broadcasts it. 
Here, we take $\Agt^a = \Agt \cup \{ T\} $. The protocol $DC^a_i(m)$ for agent $i$ is given in four steps, see Figure~\ref{fig:dca}. 
\begin{figure}[h]
$$
\begin{array}{rcllrcl} 
DC^a_i(m) &  =  & 
\begin{array}[t]{l} 
\{ i: m \rightarrow T.x_{i}  \} ; \\
\{ \} ; \\ 
\{  \} ; \\ 
\{ i: rr := y \}
\end{array} &  ({\rm for}~ i\in \Agt) ~~~~~~~~~~~ & 
DC^a_T(m) &  =  & 
\begin{array}[t]{l} 
\{ \} ; \\ 
\{ T: y:= \xor_{i\in \Agt} ~x_i \} ; \\ 
\{  T: \broadcast (y) \} ; \\ 
\{ \} 
\end{array} 
\end{array} 
$$ 
\textbf{ \caption{The abstract protocol $DC^{a}$\label{fig:dca}}}
\end{figure}
We write $DC^a(m)$ for the joint program $||_{i\in \Agt^a} DC^a_i(m) $.
Intuitively, in the abstract protocol, the agents transmit their bits across a 
secure channel to the trusted third party, who computes the exclusive-or
and broadcasts it. 

Note that since the protocol $DC^a$ makes no use of randomization, the set of worlds of the structure 
$M[DC^a(m)]$ is identical to the set of worlds of the structure $M$; only the set of defined variables
and the indistinguishability relation change. We can characterize the indistinguishability relations of 
$M[DC^a(m)]$ as follows, where we introduce the abbreviation $\xm$ for $\xor_{i\in \Agt} ~i.m$.

\begin{lemma} \label{lem:simdca} 
If $M$ is a Kripke structure at which $DC^a(m)$ is enabled, and $u,v$ are worlds of $M[DC^a(m)]$ then 
$u \sim^{M[DC^a(m)]}_i v$ iff $u\sim^M_i v$ and $\pi^M(u, \xm) = \pi^M(v,\xm)$.
\end{lemma}

The program $DC(m)$ makes use of randomization, so the structure $M[DC(m)]$ has more worlds than 
the structure $M$. More specifically, it can be seen that the worlds of  $M[DC(m)]$ have the form 
$((w, \kappa_1), \kappa_2)$, where $\kappa_1$ assigns boolean values to the variables 
$i.k_e$ for $e \in E$ and $i = e_1$, and  $\kappa_2$ assigns boolean values to the variables 
$i.k_e$ for $e \in E$ and $i = e_2$. Note that by the second step of the protocol, we always have 
$\kappa_1(e_1.k_e) = \kappa_2(e_2.k_e)$ for all $e\in E$. We may therefore abbreviate such a world to $w+\kappa$, 
where $\kappa: E\rightarrow \{0,1\}$, and 
we have 
\be 
\item $\pi^{M[DC(m)]}(w + \kappa, e_1.k_e) = \kappa(e)$, 
\item $\pi^{M[DC(m)]}(w + \kappa, e_2.k_e) = \kappa(e)$, 
\item $\pi^{M[DC(m)]}(w + \kappa, i.b) = \pi(w,i.m) \xor \xor_{e\in \keys(i)} \, \kappa(e)$, 
and 
\item $\pi^{M[DC(m)]}(w + \kappa, i.rr) = \xor_{j\in \Agt} ~\pi^{M[DC(m)]}(w+ \kappa,j.b)$. 
\ee 
Note that the final equation may be simplified as follows:  
$$
\begin{array}{rcl} 
\pi^{M[DC(m)]}(w+ \kappa,i.rr)  & =  & \xor_{j\in \Agt} ~\pi^{M[DC(m)]}(w+ \kappa,j.b) \\ 
& = &  \xor_{j\in \Agt} ~(\pi^{M[DC(m)]}(w+ \kappa,j.m) \xor \xor_{e\in keys(j)}~ \kappa(e)) \\ 
& = &  (\xor_{j\in \Agt} ~\pi^M(w,j.m))\\ 
& = &  \pi^M (w,\xm)~
\end{array} 
$$ 
where the third step follows using the fact each term $\kappa(e)$ occurs twice, 
once for $e\in keys(e_1)$ and once for $e\in \keys(e_2)$. 
Based on this representation, we can characterize the indistinguishability relations of 
$M[DC(m)]$ as follows: 

\begin{lemma} \label{lem:simdc} 
If $M$ is a Kripke structure at which $DC(m)$ is enabled, and $u+\kappa$ and $v+\lambda$ are worlds of $M[DC(m)]$ then 
$u+\kappa \sim^{M[DC(m)]}_i v+\lambda$ iff 
\be \item $u\sim^M_i v$ and 
\item $\kappa(e) = \lambda(e)$ for all $e\in \keys(i)$ and 
\item $\pi^{M}(u,j.m) \xor \xor_{e\in keys(j)} \kappa(e) = \pi^{M}(v,j.m) \xor \xor_{e\in keys(j)} \lambda(e)$ for all 
$j\in \Agt$. 
\ee 
\end{lemma} 

\section{Proof of Abstraction} \label{sec:abstraction}

The following is implicit\footnote{Chaum's result is stated probabilistically, but the proof is 
largely non-probabilistic and establishes this result.} 
in the proof of a key result concerning the DC protocol that is proved in Chaum \cite{chaum} (Section 1.4).

\begin{lemma} \label{lem:chaum} 
For all $i\in Agt$ and for all  functions  $\kappa:E \rightarrow \{0,1\}$ and $\mu:\Agt   \rightarrow \{0,1\}$ and $\mu':\Agt   \rightarrow \{0,1\}$ 
such that $\xor_{i\in \Agt} \, \mu(i) = \xor_{i\in \Agt}\,  \mu'(i)$, there exists a function $\lambda :E \rightarrow \{0,1\}$ 
such that $\kappa \upharpoonleft keys(i) = \lambda\upharpoonleft keys(i)$ and
for all $j \in \Agt$, we have  $\mu(j)  \xor \xor_{e\in keys(j)}\,  \kappa(e) = \mu'(j) \xor \xor_{e\in keys(j)} \, \lambda(e)$
\end{lemma}

Note that the variables introduced by $DC(m)$ are the variables $i.k_e$, $i.b$ and $i.rr$ for $i\in \Agt$ and $e\in E$.  
The variables introduced by $DC^a(m)$ are $T.x_i$, $T.y$ and $i.rr$ for $i\in \Agt$. 
Hence the set of variables introduced by both protocols is the set  $\{i.rr ~|~ i\in \Agt\}$. 
The following result states that these variables are introduced by these protocols in such a way as to extend a bisimulation
between given concrete and abstract structures to the new variables. 

\begin{theorem} \label{thm:abstractdc}
Suppose that  $M \approx_{V,\Agt} M^a$  for a set of variables $V$ containing all variables in the expressions $i.m$ for $i\in \Agt$, 
and let $DC(m)$ be enabled at $M$ and $DC^a(m)$ be enabled at $M^a$. 
Then $M[DC(m)] \approx_{V\cup \{i.rr \,|\, i\in \Agt\}, \Agt} M^a[DC^a(m)]$. 
\end{theorem}

\begin{proof} 
Let $M = \langle W, \Agt ,\{\sim_{i}\}_{i\in \Agt}, Prop, \pi \rangle$ and
let $$M^a = \langle W^{a}, \Agt^{a}, \{\sim^{a}_{i}\} _{i\in \Agt^a}, Prop^{a}, \pi^{a} \rangle~.$$
We write 
$$M[DC(m)]= \langle W', \Agt ,\{\sim'_{i}\}_{i\in \Agt}, Prop', \pi' \rangle$$ and
$$M^a[DC^a(m)] = \langle W^{a'}, \Agt^{a}, \{\sim^{a'}_{i}\} _{i\in \Agt^a}, Prop^{a'}, \pi^{a'} \rangle~.$$
As noted above, we have $W^{a'} = W^a$ and 
$$W' = \{w+\kappa~| ~w\in W,~\kappa : E \rightarrow \{0,1\}\}~.$$ 

Let $R\subseteq W\times W^a$ be the bisimulation relation witnessing  $M \approx_{V,\Agt} M^a$. 
We define the relation $\Re \subseteq (W' \times W^{a'})$ as follows:
$w+ \kappa \, \Re\, w'$ if $w R w'$. 
We claim that this relation witnesses $M[DC(m)] \approx_{V\cup \{i.rr \,|\, i\in \Agt\}, \Agt} M^a[DC^a(m)]$.

{\em Atoms:} We need to check that for all $v\in V\cup \{i.rr \,|\, i\in \Agt\}$, 
if $w+ \kappa \, \Re\, w'$ then $\pi'(w+\kappa,v) = \pi^{a'}(w',v)$. For propositions $v\in V$, this 
is immediate from the facts that $w+ \kappa \, \Re\, w'$ implies $w R w'$, that $R$ is a $(V,\Agt)$-bisimulation, 
and that $\pi'(w+ \kappa,v) = \pi(w,v)$ and $\pi^{a'}(w',v) = \pi^{a}(w',v)$. For the variables $i.rr$, we argue as follows. Note that since 
the variables in $i.m$ are included in $V$, it follows that $\pi'(w+\kappa,i.m) = \pi^{a'}(w',i.m)$, 
and hence that  $\pi'(w+ \kappa,\xm) = \pi^{a'}(w',\xm)$.
As noted above, we have 
$\pi'(w+ \kappa,i.rr)  =   \pi'(w+ \kappa,\xm)$. 
By the program for $DC^a(m)$, we also have 
$\pi^{a'}(w',i.rr)= \pi^{a'}(w',\xm)$. 
Combining these equations yields $\pi'(w+ \kappa,i.rr) = \pi^{a'}(w',i.rr)$. 
Thus, we have that $\Re$ preserves all propositions in $V\cup \{i.rr \,|\, i\in \Agt\}$.

\textit{Forth:} Let $i\in \Agt$, $u+\kappa, v+\lambda \in W'$, and let $u^{a'} \in W^{a'}$ 
such that  $u+\kappa \sim' _{i} v+\lambda$
and $u+\kappa \,\Re\, u^{a'}$. We need to show that there exists $v^{a'} \in W^{a'}$ 
such that $v+\lambda\, \Re \,v^{a'}$ and $u^{a'} \sim^{a'}_{i} v^{a'}$. 
We argue as follows. From $u+\kappa ~ \Re ~ u^{a'}$ it follows that 
$u R u^{a'}$.  Also,  from $u+\kappa \sim'_{i} v+\lambda$ it follows 
by Lemma~\ref{lem:simdc} that $u \sim_{i} v$.
Since $R$ is a bisimulation, we obtain that there exists a world $v^a \in W^a$
such that $u^{a'} \sim^{a}_{i} v^a$ and $v R v^{a}$. Since $W^{a'} = W^{a}$ we may 
define $v^{a'}$ to be $v^a$. It is immediate from the definition of $\Re$ 
and the fact that $v R v^{a'}$ that $v +\lambda \Re v^{a'}$. 
To show   $u^{a'} \sim^{a'}_{i} v^{a'}$ we use the characterization of  $\sim^{a'}_{i} $
of Lemma~\ref{lem:simdca}. We already have that $u^{a'} \sim_{i} v^{a'}$ by construction, 
so it remains to show $\pi^a(u^{a'},\xm) = \pi^a(v^{a'},\xm)$. 

From the fact that 
$v R v^{a'}$, and that all variables in $i.m$ are in $V$, 
we have that $\pi(v,\xm) = \pi^{a}(v^{a'},\xm)$.
Similarly, from $u R u^{a'}$, 
we have that $\pi(u,\xm) = \pi^a(u^{a'},\xm)$. 
Further, since  $u+\kappa \sim'_{i} v+\lambda$, it follows 
by Lemma~\ref{lem:simdc} that  $\pi(u,\xm) = \pi(v,\xm)$. 
Combining these equations yields $\pi^a(u^{a'}\xm) = \pi^{a}(v^{a'},\xm)$, 
giving the remainder of what we require for the conclusion that 
$u^{a'} \sim^{a'}_{i} v^{a'}$.

\textit{Back:}  Let $i\in \Agt$, $u+\kappa \in W'$, and let $u^{a'}, v^{a'}  \in W^{a'}$ 
such that $u+\kappa \,\Re\, u^{a'}$ and  $u^{a'} \sim^{a'} _{i} v^{a'}$. 
We need to show that there exists $v + \lambda \in W'$ 
such that $u+ \kappa \sim'_i v+\lambda$ and $v+\lambda \, \Re \,v^{a'}$.
We identify the world $v\in W$ as follows. 
From $u+\kappa \,\Re\, u^{a'}$ we have  that $u R u^{a'}$ and 
from  $u^{a'} \sim^{a'} _{i} v^{a'}$ we have (by Lemma~\ref{lem:simdc})
that $u^{a'} \sim^{a} _{i} v^{a'}$. Since $R$ is a bisimulation, there exists a 
value $v\in W$ such that $u \sim_i v$  and $v R v^{a'}$. 

From $u^{a'} \sim^{a'} _{i} v^{a'}$ and Lemma~\ref{lem:simdca}, we obtain that 
$\pi^{a'}(u^{a'},\xm) = \pi^{a'}(v^{a'},\xm)$, 
hence also $\pi^{a}(u^{a'},\xm) = \pi^{a}(v^{a'},\xm)$. 
From the fact that $R$ is a bisimulation preserving the propositions $V$, 
we get from  $u R u^{a'}$ and  $v R v^{a'}$ that 
$\pi(u,\xm) = \pi^{a}(u^{a'},\xm)$ and $\pi(v,\xm) = \pi^{a}(v^{a'},\xm)$. 
Combining these equations yields $\pi(u,\xm) = \pi(v,\xm)$. 

Note that  $v R v^{a'}$  implies that $v+ \lambda \, \Re\, v^{a'}$ for all 
$\lambda: E \rightarrow \{0,1\}$, giving half of what we require. 
It therefore remains to find a value of $\lambda$ such that 
$u+ \kappa \sim'_i v+\lambda$. Since we already have $u\sim_i v$,  
 this requires, by Lemma~\ref{lem:simdc}, that we find $\lambda$ such that 
$\kappa(e) = \lambda(e)$ for all $e\in \keys(i)$ and 
$\pi^{M}(u,j.m) \xor \xor_{e\in keys(j)} \kappa(e) = \pi^{M}(v,j.m) \xor \xor_{e\in keys(j)} \lambda(e)$ for all 
$j\in \Agt$.  Since $\pi(u,\xm) = \pi(v,\xm)$, the existence of such a function $\lambda$ is guaranteed by Lemma~\ref{lem:chaum}, 
on taking $\mu(i) = \pi(u,i.m)$ and $\mu'(i) = \pi(v,i.m)$. 
\qed
\end{proof}

This result gives us that, modulo bisimulation, the programs $DC(m)$ and $DC^a(m)$ have the
same effect on the agent's mutual states of knowledge. We have a similar result if we consider 
the effect of joint actions $A$: 

\begin{lemma} \label{lem:actions} 
Let $M$ and $M'$ be Kripke structures such that $M\approx_{V,\Agt} M'$, and 
let $A$ be a joint action, writing variables $V_A$, such that $A$ is enabled at both $M$ and $M'$. 
Then $M[A] \approx_{V\cup V_A,\Agt} M'[A]$,
\end{lemma} 
\begin{proof} 
Suppose $R$ is a bisimulation witnessing $M\approx_{V,\Agt} M'$, 
and we represent the worlds of $M[A]$
as $w+\kappa$ where $w$ is a  world of $M$ 
and $\kappa: V_A \rightarrow \{0,1\}$, where 
$\pi^{M[A]}(w + \kappa,v) =  \kappa(v)$ for $v\in V_A$. 
(This requires some constraints on the set of $w+\kappa$, to handle the case of 
variables $v\in V_A$ that are not written by $i:\rand(v)$ statements.) 
The worlds of $M'[A]$ may be similarly represented as $w+\kappa$ where $w$ is a world of $M'$. 

Then it is easily shown that the  relation $r'$ defined by 
$u+\kappa \,R'\, v+\lambda$ if $uRv$ and $\kappa = \lambda$ 
is a bisimulation.  \qed
\end{proof}

Combining Theorem~\ref{thm:abstractdc} and Lemma~\ref{lem:actions}, we obtain the following
by a straightforward induction.
(Note that we 
use fresh variables  
$k_e, b, rr, x_i$ and $y$ in each of the instances of $DC_i$ and $DC^a_i$.)

\begin{theorem} \label{thm:main} 
Let $M$ and $M^a$ be Kripke structures with $M\approx_{V,\Agt} M^a$, and 
let   
$$\begin{array}{c}
P= Q_1; DC(m_1); Q_2 ;  DC(m_2); \ldots  DC(m_k); Q_{k+1}~~{\it and} \\[5pt] 
P^a = Q_1; DC^a(m_1); Q_2 ;  DC^a(m_2); \ldots  DC^a(m_k); Q_{k+1}
\end{array}$$
where the $Q_i$ are programs involving agents $\Agt$. 
Let  $V'$ be the set of all variables written by the programs $Q_i$, 
as well as the variables $i.rr$ introduced by the DC instances.  Assume that 
the $Q_j$ and $m_j$ read only variables from $V\cup V'$. 
Then if $P$ is enabled at $M$, and $P^a$ writes no variable in $M^a$,  then 
$P^a$ is enabled at $M^a$ and $M[P]\approx_{V\cup V',\Agt} M^a[P^a]$. 
\end{theorem} 

This result states that if we have a complex protocol $P$, constructed by using 
multiple instances of the DC protocol interleaved with other actions, 
then we abstract $P$ by abstracting each of the instances of $DC$ to $DC^a$, 
while preserving the truth values of all epistemic formulas.  This enables 
optimization of model checking epistemic formulas in $M[P]$ by 
applying model checking to $M[P^a]$ instead.  (Note that always $M \approx M$.)

\section{The Two-phase Anonymous Broadcast Protocol}\label{sec:rounds} 

As noted above, the basic version of the Dining Cryptographers protocol enables a
signal to be anonymously transmitted under the assumption that at most one 
agent wishes to transmit.  One of Chaum's considerations is the use of the protocol for
more general anonymous broadcast applications, and he writes: 
\begin{quote} 
The cryptographers become intrigued with the ability to make messages
public untraceably.  They devise a way to do this at the table for a
statement of arbitrary length: the basic protocol is repeated over and
over; when one cryptographer wishes to make a message public, he merely
begins inverting his statements in those rounds corresponding to 1's in
a binary coded version of his message. If he notices that his message
would collide with some other message, he may for example wait for a
number of rounds chosen at random from some suitable distribution
before trying to transmit again.
\end{quote}  

As a particular
realization of this idea, he discusses grouping communication into blocks and 
the use of the following 
{\em two-phase broadcast} protocol using {\em slot-reservation}:  
\begin{quote} 
In a network with many messages per block, a first block may be used by various anonymous senders 
to request a ``slot reservation'' in a second block. A simple scheme would be for each 
anonymous sender to invert one randomly selected bit in the first block for each slot they wish to reserve
in the second block. After the result of the first block becomes known, the participant 
who caused the ith bit in the first block sends in the ith slot of the second block. 
\end{quote} 
This idea has been implemented as part of the Herbivore
system\cite{Herbivore}.  

Chaum's  discussion leaves open a number of questions concerning the
protocol.  For example, what exact test is applied to
determine whether there is a collision? 
Which agents are able to detect a collision? 
Are there situations where some agent expects
to receive a message, but a collision occurs that it 
does not detect (although some other agent may do so?) 
Under what exact circumstances does an agent know that 
some agent has sent a message? When can a sender be 
assured that all others have received the message? 

In previous work, we have studied such questions in a 
3-agent version of the protocol \cite{AlBatainehMeyden10}. 
Our approach was to model the protocol as a knowledge-based program 
and to use epistemic model checking as a tool to help us identity precisely 
the conditions under which an agent obtains some types of knowledge of interest. 
The approach helped us to identify some unexpected situations in which 
relevant knowledge is obtained. 
We recap the definition of knowledge-based programs and our
formulation of the 2-phase protocol as a knowledge-based program 
in the following sections, after which we study this knowledge-based program 
further using the abstraction developed above. 

\section{Implementation of Knowledge-based Programs} \label{sec:kbp} 

Knowledge-based programs \cite{FHMVbook} are like standard programs, except that 
expressions may refer to an agent's knowledge. 
That is, in a knowledge-based 
program for agent $i$, we may find 
statements 
 of the form
 ``$v:= \phi$", 
where $\phi$ is a formula of the logic of knowledge, i.e., 
a boolean combination of atomic formulas concerning the agent's observable variables
and formulas of the form $K_i\psi$.

Unlike standard programs, knowledge-based programs cannot in general be directly executed,
since the satisfaction of the knowledge subformulas depends on the
set of all runs of the program,  which
in turn depends on the satisfaction of these knowledge subformulas. 
This apparent circularity is handled by treating a knowledge-based program as a  specification, and 
defining when a concrete standard program satisfies this specification. 
We give a formulation of the semantics of knowledge-based programs tailored to the 
programming language of the present paper.

Suppose that we have a concrete program $P$ of the same syntactic structure as the knowledge-based program ${\bf P}$, 
in which each knowledge-based expression 
$\phi$ is replaced by a concrete predicate $p_\phi$ of the local variables of the agent. 
Starting at an initial Kripke Structure $M_0$, 
the concrete program $P$ generates a set of runs that form the worlds of a
Kripke Structure $M_0[P]$. 
We now say that $P$ is an {\em implementation of the knowledge-based program ${\bf P}$ from $M_0$}  if
for each joint action $A$ in the program $P$, corresponding to a joint action ${\bf A}$ in the 
knowledge-based program,  if we write $P = P_0;A;P_1$, where 
$P_0$ and $P_1$ are programs, then for each knowledge condition  $\phi$ occurring in ${\bf A}$, 
we have  $M_0[P_0]\models p_\phi \dimp \phi$.
That is, the concrete condition is equivalent to the knowledge condition in the implementation
at each point in the program where it is used. 
(In a more general formulation, where knowledge conditions may contain temporal operators,  
knowledge-based programs may have no implementations, a behaviourally unique implementation, or
many implementations, but for the restricted language we consider it can be shown that there is a unique implementation.) 

We now describe a partially automated process, using epistemic model checking, 
that can be followed to find implementations of knowledge-based programs ${\bf P}$. 
The user begins by introducing a local boolean variable $v_\phi$ for each knowledge formula
$\phi= K_i\psi$ in the knowledge-based program, and replacing $\phi$ by $v_\phi$. 
Treating $v_\phi$  as a ``history variable'', the user may also add to the program statements of the form $v_\phi:=e$, 
relying on their intuitions concerning situations under which the epistemic formula $\phi$ will be true. This produces a
standard program $P$ that is a candidate to be an implementation of the knowledge-based program
${\bf P}$. 
(It has, at least,  the correct syntactic structure.) 
To verify  the correctness of $P$ as an implementation of ${\bf P}$, the user must now 
check that the variables $v_\phi$ are being maintained so as to be equivalent to the 
knowledge formulas that they are intended to express. This can be done using epistemic model checking, 
where we verify formulas of the form $v_\phi \dimp K_i\psi$
at points in the program where the condition $\phi$ is used.  

In general, the user's guess concerning the concrete condition that is equivalent to the  knowledge formula 
may be incorrect, and the 
model checker will report the error. In this case, the model checker can be used  to generate an {\em error trace}, a 
partial run leading to a situation that falsifies the formula being checked. 
The next step of our process requires the user to analyse this error trace (by inspection and human reasoning) 
in order to understand the source of the error in their guess for the concrete condition representing the knowledge formula. 
As a result of this analysis, a correction of the assignment(s) to the variable $v_\phi$ is made 
by the user (this step may require some ingenuity on the part of the user.) 
The model checker is then invoked again to check the new guess. 
This process is iterated until a guess is produced for which all the formulas of interest are found to be true, 
at which point an implementation of the knowledge-based program has been found. 
We refer the reader to our previous work \cite{AlBatainehMeyden10} for further discussion and examples of the 
application of this iterative process. (We deemphasize the process in the present paper, and 
focus on the results.)

\section{The Two-phase Broadcast Protocol as a Knowledge-based Program}\label{sec:dckbp}

We now 
give a formulation 
of Chaum's two-phase protocol (see Section~\ref{sec:rounds}) 
as a knowledge-based program, and discuss the associated verification conditions. 
(The knowledge-based program is similar to that given in our earlier work, but 
includes some improvements.) 

We assume that there are $n$ agents, and $\Agt= \{1..n\}$. 
Figure~\ref{fig:kbp} represents the 2-phase protocol by giving a knowledge-based program 
for agent $i$. 
The local variable \verb+slot-request+, assumed to be defined 
in the structure from which the program is run,  records the slot number 
(in the range 1..n)  that this  agent will attempt to reserve. If \verb+slot-request+=0, then the 
agent will not attempt to reserve any slot.  The variable
\verb+message+, also assumed to be defined, records the single bit message that the agent 
wishes to anonymously broadcast (if any). 
The program introduces the 
variables $\rcvd0$ and $\rcvd1$, 
as well as a variable $\dlvrd$. (Additional new variables, 
are implicit in the instances of DC$_i$.) 

\begin{figure} [h]
{\bf P}$_i$ = 
 \{\\
local variables:\\ 
\hspace*{20pt}  \slotrequest : [0..n],\\
\hspace*{20pt}  \msg : Bool,\\
\hspace*{20pt}  \rcvd0, \rcvd1, \dlvrd : Bool;\\
//reservation phase\\
\verb+for+ ($s =1$; $s \leq n$; $s$++)\\ 
  $\{$\\
\hspace*{5pt}   DC$_i$(\slotrequest=$s$);\\
  $\}$ \\
//transmission phase\\
\verb+for+ ($s =1$; $s \leq n$; $s$++) \\
  $\{$ \\
\hspace*{5pt} DC$_i$(\verb+if+ ($\slotrequest=s$ $\land$ $\neg K _{i}(\conflict(s))$\\
\hspace*{20pt} \verb+then+  \msg \\
\hspace*{20pt} \verb+else+   false) );\\
  $\}$ \\
 \rcvd0:= $K_i(\sender(i,0))$; \\ 
 \rcvd1 := $K_i(\sender(i,1))$;  \\ 
\dlvrd := $
\begin{array}[t]{r}
\bigwedge_{x\in Bool } ((\mbox{$\msg = x$} ~ \land ~\mbox{$\slotrequest \neq 0$}) \rimp\\  K_i ( \bigwedge_{j\neq i} K_j\sender(j,x)))
 \end{array}
 $\\
\}
\textbf{ \caption{ \label{fig:kbp}The knowledge-based program $CDC$}}
\end{figure}

The term $\conflict(s)$ in the knowledge-based program 
represents that there is a conflict on slot $s$. 
This is a global condition that is defined as
$$ \conflict(s) = \bigvee_{i\neq j} (i\verb+.slot-request+= s= j\verb+.slot-request+)~.$$
i.e., there exist two distinct agents $i$ and $j$ both requesting slot $s$. 

The term $\sender(i,x)$ represents that an agent 
is sending message $x$.
Thus, the variable 
\verb+rcvd0+ is assigned to be true if the agent  learns that someone 
is trying to send the bit 0, and similarly for \verb+rcvd1[s]+. 
However, there are some subtleties in the implementation that lead us to consider 
two distinct versions of the program. 
In one version, called {\em strong reception}, we use the definition  
$$ \sender(i,x) = \bigvee_{j\neq i} (j.{\tt message} = x \land j.\slotrequest \neq 0)~.$$
That is, we take an agent to have received the bit $0$ if it knows that {\em some other} 
agent is sending the message $x$. In the other, that we refer to as {\em weak reception}, 
we define 
$$ \sender(i,x) = \bigvee_{j} (j.{\tt message} = x \land j.\slotrequest \neq 0)~.$$
That is, we take an agent to have received the bit $0$ if it knows that {\em some} 
agent is sending the message $x$, possibly itself. Since an agent always knows 
its own message $x$, it trivially knows $ \sender(i,x)$   if it is trying to send a message (i.e., $i.\slotrequest \neq 0$), 
so this may seem very weak. However, since other agents may consider it possible that the agent is 
not seeking to send a message, we see that it becomes of greater interest in the context of
an agent's knowledge of delivery of its message, represented by 
the assignment for the variable $\dlvrd$. 

We note that this representation of the 2-phase protocol as a 
knowledge-based program 
is {\em speculative}: an agent transmits in a slot so long as it does not know 
that there is a conflict. This allows that a collision will occur during the transmission phase. 
   
Since an agent may attempt to reserve a slot, and then back off, 
or may send in a reserved slot without success because of a collision during the transmission phase, the protocol does not 
guarantee that the message will be delivered. In this case, the 
agent is required to retry the transmission in the next run of the protocol. 
So  that  it can determine whether a retry is necessary, the final assignment 
to the variable $\dlvrd$ captures  whether the agent knows that 
its (anonymous) transmission has been successful. this assignment captures that 
the transmission is successful if the agent knows that the other agents
know that some agent is sending its message.  
We similarly refer to {\em weak delivery} and {\em strong delivery}  
depending on which version of the predicate $ \sender(i,x)$ is used.%
\footnote{We remark that in case of weak delivery, replacing the expression $\bigwedge_{j\neq i} K_j\sender(j,x)$
by $\bigwedge_{j} K_j\sender(j,x)$ in the assignment to $\dlvrd$ would have no effect, since in the weak case it always holds 
that $(i.\msg = x ~ \land ~i.\slotrequest \neq 0 )\rimp K_i (\sender(i,x)) $.
} 

We remark that the knowledge-based program  is interpreted with respect to 
the assumption of perfect recall, and implementations may make use of of history variables 
to capture observations that the agent makes during the running of the protocol. 
Thus, by placing the reception and delivery assignments at the end of the program (rather than just after each DC instance), 
we ensure that the agents are able to behave optimally by making use of all information they gather during the running of the program. 
As we discuss below, this allows us to capture some subtle sources of information. 

In Figure~\ref{fig:imp}, we give the generic structure of a possible implementation of the 
knowledge-based program, as we seek using our partially-automated process. 
The variable \verb+kc[s]+ is used to represent the epistemic condition
concerning conflict  in the knowledge-based program (i.e., $\neg K _{i}(\conflict(s))$). 
Thus, in verifying that we have an implementation, the 
key condition to be checked is whether $\verb+kc[s]+\dimp \neg K _{i}(\conflict(s))$ 
just after this variable is  assigned. 
The main difficulty in finding an implementation is to find the appropriate  concrete 
assignment (to take the place of the ``???") for this variable that will make this condition valid. 
Similarly we seek assignments to the variables \verb+rcvd0[s], recvd1[s]+ that give these the intended meaning.

\begin{figure} [h]
$P_i$ =  \{ \\
local variables:\\ 
\hspace*{20pt}  \slotrequest: [0..n],\\
\hspace*{20pt}  \msg: Bool,\\
\hspace*{20pt}  \rcvd0, \rcvd1, \dlvrd: Bool,\\
\hspace*{20pt}  \kc[n]: Bool;\\
//reservation phase\\
\verb+for+ ($s =1$; $s \leq 3$; $s$++)\\ 
  $\{$\\
  \hspace*{5pt}   DC$_i$(\slotrequest == $s$); \\
  $\}$ \\
//transmission phase\\
\verb+for+ ($s =1$; $s \leq n$; $s$++) \\
  $\{$ \\
\hspace*{5pt} \kc[s] :=???; \\
\hspace*{5pt} DC$_i$(\verb+if+ (\slotrequest == s $\land$ \kc[s])\\
\hspace*{20pt}   \verb+then+  \msg \\
\hspace*{20pt}   \verb+else+  false);\\
  $\}$\\
 \rcvd0 := ???; \\ 
 \rcvd1 := ???;\\ 
  \dlvrd := ??? \\ 
  \}
\textbf{ \caption{\label{fig:imp}A generic implementation of  CDC}}
\end{figure}

We note that each of the instances of the protocol $DC_i$ introduces additional 
variables, which may be used in the concrete predicates we substitute for the 
``???". In particular, they introduce round result variables, which we denote by 
\verb+rr[t]+  for $t\in \{1..2n\}$. Here $rr[t]$ represents the round result variable
from the $t$-th instance of $DC_i$ in the implementation.  The implementations
also introduce key variables $k_e$ and $b$, which need to be 
separated in the different instances: we may similarly use $k_e[t]$ and $b[t]$
to denote the $t$-th instance of such a variable. 

We now discuss the formulas that are used to verify the implementation.  
As discussed above, these conditions need to be verified at specific stages of the program, 
viz., the step before the occurrence of the knowledge formula of interest. 

The first formula of interest concerns the correctness of the 
guess for the knowledge condition $\neg K_i(\conflict(s))$ (in case of the 
speculative implementation, or $K_i(\neg \conflict(s))$ (in the case of the 
conservative implementation). In the implementation, 
this condition is represented by the variable \verb+kc[s]+. 

\textit{Specification 1:
{\tt kc[s]}  correctly represents knowledge of the existence of a conflict in slot $s=1..3$. }
$$ i.\verb+kc[s]+ \dimp \neg K_{i} (\conflict(s))~~~~~~~(1)$$

Next, the protocol has some positive goals, viz., to allow agents to broadcast some 
information, and to do so anonymously. Successful reception of a bit 
 is intended to be represented by the variables 
 \verb+rcvd0+ and \verb+rcvd1+. 
To ensure that the assignments to these variables correctly implement their
intended meaning in the knowledge-based program, we use specifications of the
following form. 

\textit{Specification 2: reception variables  correctly represent transmissions by others} 

$$ i.\rcvd0 \dimp K_i(\sender(i,0) ~~~~~~~~~(2a)$$ 
and 
$$ i.\rcvd1 \dimp K_i(\sender(i,1)) ~~~~~~~~~ (2b)$$ 

Similarly, we need to verify correct implementation of the agent's knowledge about whether its transmission is 
successful.  

\textit{Specification 3: delivery variables  correctly represent knowledge about delivery} 
$$
\begin{array}{l} 
 i.\dlvrd \dimp 
\bigwedge_{x\in Bool} (i.\msg = x  \land i.\slotrequest \neq 0\\
~~~~~~~~~~~~~~~~~~~~~~~~~ \rimp K_i ( \bigwedge_{j\neq i} K_j\sender(j,x)))
\end{array} 
$$

There are strong and weak versions of Specifications 2 and 3, depending on the choice for $\sender(i,x)$. 

Finally, the aim of the protocol is to ensure that when information is transmitted, this 
is done anonymously. An agent may know that one of the other two agents 
has a particular message value, but it may not know what that value is
for a specific agent. We may write the fact that agent $i$ knows the 
value of a boolean variable $x$ by the notation $\hat{K}_i(x)$, defined 
by $\hat{K}_i(x) = K_i(x) \lor K_i(\neg x)~.$
Using this, we might first attempt to specify anonymity as
$\bigwedge_{j\neq i} ( \neg \hat{K}_i(j.\msg))$, i.e., agent $i$ knows no other's message. 
Unfortunately, the protocol cannot be expected to satisfy this: suppose that 
all agents manage to broadcast their message and all messages have the same value
$x$: then each knows that the other's value is $x$. We therefore write the 
following weaker specification of anonymity: 

\textit{Specification 4:  The  protocol preserves anonymity}
$$  \bigvee_{x=0,1} K_i( \bigwedge_{j\neq i} (j.\msg = x)) \lor 
\bigwedge_{j\neq i} ( \neg \hat{K}_i(j.\msg))~.$$ 
This is checked at the very end of the protocol.

\section{Model Checking Performance}\label{sec:results} 

To verify the specifications for the knowledge-based program in a putative implementation, 
we have applied  the epistemic model checker MCK \cite{mck}. We refer the reader to 
our previous work \cite{AlBatainehMeyden10} for a description of some 
of the particularities of how this is done. Since the details are straightforward, 
we focus here on how the abstraction developed in this paper 
impacts the performance of model checking. 

We would like to verify whether a putative implementation $P$ implements the knowledge-based program ${\bf P}$ 
from an initial structure $M_0$. This requires that we model check the formulas from the previous section. 
Since these formulas concern only the initial variables of the agents, and variables introduced outside
the scope of the $DC_i$ calls, it follows from Theorem~\ref{thm:main} that we may 
verify instead whether these formulas hold at appropriate times during the running of the 
abstract program $P^a$ that we obtain by replacing each instance of $DC_i$ in $P$ by $DC_i^a$. 

We have performed some experiments in which we use MCK for this model checking problem. 
MCK is a symbolic model checker, and model checking a formula 
involves first building a symbolic (Binary Decision Diagram \cite{Clarke99}) representation of the 
model itself, and then using this representation in the construction of a symbolic 
representation of the situations where the particular formula of interest is false.  
All specifications are checked using  the perfect recall interpretation of knowledge
and the model checking algorithm for this semantics which is described in \cite{MS}
(which is flagged by {\tt spec\_spr\_xn} in MCK). 
To estimate individual formula timings, we deduct model construction times 
(estimated by the time to model check  the specification True), from the actual time for model checking each specification 
(which includes model construction and formula verification time.)  
All experiments are conducted on a PC with Intel(R) Xeon(R) 4 $\times$ 3 GHZ, and 16 GB memory,
using 
MCK 0.1.1. Where the execution crashed due to a memory error, we report ``x" in the tables.
  
Our methodology for identifying an implementation of the knowledge-based program 
requires that we perform model checking on number of different approximations to the 
final implementation, and, when a specification fails, using the counter-example found to revise the approximation. 
Table~\ref{table:results-initial} gives the runtimes for the initial program, in which we 
guess the predicate False for the implementation of all knowledge formulas in the
knowledge-based program. For each specification $x$ we give runtimes 
for model checking the specification in the concrete program and the abstract program (indicated by $x^a$). 
We count the cost of verifying all instances of the specification required to check the 
correctness of the implementation at different times where the knowledge condition 
occurs in the program. (With $n$ agents, we need to check 
Specification 1 at $n$ locations in the implementation, but specifications  2-4 just once.)  
As we improve the approximation, the program becomes more complex, 
and the model checking runtimes increase. In Table ~\ref{table:results-final} we give the 
runtimes for the final approximation, in which we have identified a program that is verified
as implementing the knowledge-based program. 
 
\begin{table}[h]
\begin{center} 
\begin{tabular}{|c||c|c||c|c||c|c||c|c||c|c|} 
\hline 
\multicolumn{3}{|c|}{} & \multicolumn{8}{c|}{Specification} \\
\hline 
$n$  & Model & Model$^a$ & 1 & 1$^a$ & 2 & 2$^a$ & 3 & 3$^a$ &  4 & 4$^a$  \\
\hline 
3 & 0.4 & 0.24& 43& 5 & 5880 & 41 & 6100 & 4 & 6300 & 5 \\ 
4 & 29.15& 4.2 & x & 34 &x & 68 & x & 69 & x & 70 \\ 
5 & x& 63& x& 4800& x&5400 & x& 5500& x& 5544\\  
\hline 
\end{tabular} 
\end{center} 
\caption{Model Checking Runtimes (seconds)-- initial approximation \label{table:results-initial}}
\end{table} 
\begin{table}[h]
\begin{center} 
\begin{tabular}{|c||c|c||c|c||c|c||c|c||c|c|} 
\hline 
\multicolumn{3}{|c|}{} & \multicolumn{8}{c|}{Specification} \\
\hline 
$n$  & Model & Model$^a$ & 1 & 1$^a$ & 2 & 2$^a$ & 3 & 3$^a$ &  4 & 4$^a$ \\
\hline 
3 & 0.45 & 0.4& 50 & 16 & 7200 &127 & 7350& 34& 7400& 18\\ 
4 & 135& 6& x &167& x & 378& x & 251& x & 252 \\ 
5 &x & 74& x& 1096 & x& 1957 & x& 1979&x &1998 \\  
\hline 
\end{tabular} 
\end{center} 
\caption{Model Checking Runtimes (seconds) -- final implementation \label{table:results-final}}
\end{table} 

For a more detailed indication of the impact of the abstraction, 
Table~\ref{table:results-rounds} compares the runtimes for model checking the anonymity specification  
(Specification 4) in the concrete and abstract programs for the final implementation after a given number of rounds
of the Dining Cryptographers Protocol. Note that the maximum number of rounds of Dining Cryptographers 
in the 2-phase protocol is twice the number of agents. 

\begin{table}[h]
\begin{center} 
\begin{tabular}{|c|c||c|c|c|c|c|c|c|c|c|c|} 
\hline 
\multirow{2}{1cm}{Agents}& \multirow{2}{1cm}{version}& \multicolumn{10}{c|}{Rounds} \\  
 & & 1 & 2 & 3 & 4 & 5 & 6 & 7 & 8 & 9 & 10 \\ 
\hline 
 3 & concrete & 0.6&0.9& 2.2& 18&335&7350&- &- &- & - \\ 
  3 & abstract &0.5& 0.6&0.7&1.6& 3.1& 17.8&- &- &- & - \\ 
  \hline
 4 & concrete  &340& 575&587&1478&2661&x& x& x&- & - \\ 
  4 & abstract &9&11&11.2&11.7& 32& 85& 86 &249 &- & - \\ 
  \hline 
   5 & concrete &x&x&x&x&x&x& x& x& x&x \\ 
 5 & abstract   &91& 110& 133& 134& 183& 311& 752& 722& 950& 1990\\ 
\hline 
\end{tabular} 
\end{center} 
\caption{Model Checking Runtimes (seconds) for Specification 4 \label{table:results-rounds}}
\end{table} 

In all these experiments, the runtimes obtained indicate that the abstraction results in 
a significant decrease of runtimes, (in some cases of several orders of magnitude) 
and helps to bring problems of larger scale (in particular, with larger numbers of agents 
and greater numbers of rounds of the basic Dining Cryptographers protocol) within the bounds of 
feasibility of model checking.  

\section{Implementations of  the knowledge-based program}\label{sec:imp-results} 

Using the optimization obtained from the abstraction, we have been able to extend 
our previous analysis of the knowledge-based program in the 3-agent case 
to the cases of 4 and 5 agents, gaining more insight into the $n$-agent case for general $n$. 
We now describe the implementations we found for the program, 
which demonstrate that the protocol contains some further subtle flows
of information beyond those we found in the 3 agent case. 

One point worth noting is that, in addition to providing an optimization of 
epistemic model checking, our abstraction result also provides information that is useful in 
the search for an implementation of the knowledge-based program. Observe that the 
variables  $k_e$ do not occur in the abstract version of the protocol, 
nor in the formulas we need to check to verify an implementation. 
Thus, in guessing a concrete predicate to be substituted for one of the 
knowledge conditions, we can confine our attention to predicates that 
do not contain the $k_e$ variables. Indeed, since $i.b$ is computed from 
information already at agent $i$'s disposal, we need only consider 
predicates based on agent $i$'s initial information and the round result
variables $rr[k]$. 

The first knowledge condition we need to implement, for Specification 1, 
is  $\neg K_i \conflict(s)$. Plainly, one situation 
where an agent knows that there is a conflict is when it 
attempts to reserve a slot and the round result for the reservation 
is not $1$. (So an even number of agents attempted to reserve the slot.) 
Thus, one potential implementation for  $\neg K_i \conflict(s)$ is the assignment 
 $kc[s] := \neg (\slotrequest =s \land rr[s] = 0)$. Model checking Specification 1 for this 
predicate at the point of the $s$-th transmission  confirms 
in all of the cases $n=3,4,5$ that this captures the knowledge condition 
$\neg K_i \conflict(s)$ exactly at this  point: there are no other ways 
that the agent can know of a conflict on a slot before transmitting on it, 
besides seeing a reservation clash. 
(In particular, previous transmissions do not contain any relevant information.)   

It is interesting to consider not just the knowledge condition 
$\neg K_i \conflict(s)$ that occurs in the program, but also the 
stronger condition $K_i \neg \conflict(s)$ (the formula $K_i \neg p \rimp \neg K_i p$ is a 
validity of the logic of knowledge). For example, if an agent who is broadcasting 
on a slot knows that all other agents know the slot is conflict free, then 
it knows that its message will be delivered. Thus, we have also added a local variable 
$\conflictfree(s)$ to the implementation, for $s= 1\ldots n$,  and 
and sought assignments to this variable that satisfy the formula 
$i.\conflictfree(s) \dimp K_i \neg \conflict(s)$. This turns out to be 
quite a subtle matter. 

To express this condition, it is useful to introduce a formula $C_0 = x$ where $x\in \{0,\ldots,n\}$ 
to express that the number of $0$'s obtained as round results in the reservation phase 
is $x$. We may then note the following situations in the protocol in which $K_i \neg \conflict(s)$
holds. 
\begin{itemize}
\item If $C_0 = 0$ or $C_0 = 1$, 
then the agent knows there is no conflict on any slot. Note that in this case there are at least $n-1$ agents 
who are requesting the at least $n-1$ distinct slots with reservation round result $1$, leaving at most one further agent. 
If this agent had requested any of the slots with round result $1$, this would have caused a 2-way reservation clash, 
contradicting the observed round result of $1$. Hence this agent did not request any slot, 
and {\em all} slots are conflict-free. 

\item If $C_0 \geq 2$, then in general, an agent cannot determine
whether or not there is a conflict on any of the reserved slots, 
since there may be a 3-way clash on one of these slots. 
However, in the particular case where $C_0 = 2$ and the agent itself does not request any slot 
(\verb+slot-request =0+) then $n-2$ agents are accounted for by the $n-2$ slots on which 
we see a reservation round result of $1$, and the remaining one agent cannot be assigned 
to ay slot without changing the round result, and hence the count. Hence this 
agent cannot be requesting a slot, and the agent knows that 
all slots are conflict-free.

\item  
Note that if  $C_0 = 2$ or $C_0=3$, and the agent requests a slot but detects a collision at slot reservation time, then 
there must have been at least 2 agents requesting this slot, leaving at most $n-2$ agents for the $n-1$ other slots,
where we see either $n-3$ or $n-4$ slots with reservation result of $1$. This means at least 
$n-1$ or $n-2$ agents are accounted for in total, so the number of agents remaining to contribute to a further 
collision on the remaining $n-1$ other slots is at most 1. This agent can not be assigned to any slot 
without changing the round result for that slot, so it must not be  requesting a slot. Thus, 
all the other $n-1$ slots are collision free.  

\item
The above cases use information from the reservation phase. 
Agents may also be able to deduce that slots are conflict-free as a result of 
information they obtain during the transmission phase.
 If $C_0 = 2$ or $C_0=3$, the agent requests a slot and obtains a reservation round result of 
$1$ for this slot, but then detects a collision at transmission time, then there must have been at least a 
3-way collision on that agent's slot, and by a similar argument to the 
previous case, we deduce that all the other slots are collision free.  

\end{itemize}
These conditions may be captured by the assignment 
 $$
\begin{array}[t]{l}
i. \conflictfree(s)  := C_{0} = 0 \lor C_{0} = 1 \lor (C_0 = 2 \land i.\slotrequest = 0) \lor \\ 
~~~~~((C_0 = 2 \lor C_0 = 3) \land \bigvee_{t=1}^n (s\neq t \land i.\slotrequest = t \land \rr[t] = 0 )) \lor  \\ 
~~~~~((C_0 = 2 \lor C_0 = 3) \land \bigvee_{t=1}^n (s\neq t \land i.\slotrequest = t \land \rr[t] =1 \\
\hspace{2in} \land \rr[n+t] \neq i.\msg))
\end{array} 
 $$
The above formula states several concrete conditions under which the agent
knows there is no conflict on a particular slot $s$. We have verified 
by model checking that for $n=3$, 4, and 5 that, at the end of the protocol, for all slots $s$  
we have $i. \conflictfree(s)  \dimp K_i \neg \conflict(s)$, and conjecture that 
it holds for all $n$.  

We remark that in the case of $C_0 = 0$ or $C_0 = 1$, this information is available to all agents, and it is 
common knowledge\footnote{A fact is common knowledge \cite{HM90} if all agents know it, all agents know that all other agents know it, 
and so on for all levels of iteration of knowledge.}  that all slots are conflict free.  In the other cases, collision 
freedom on a slot may be known to some agents but not to others. For example, consider the 
situation with $n=4$ and where the $\slotrequest$ and $\msg$ values and round results are given
as in Figure~\ref{fig:cfex}(a). 
Here agent 2 sees a reservation collision and two 1's elsewhere, 
so knows that slots 1 and 4 are collision free. However, agent 1 does not know this, 
since the scenario of Figure~\ref{fig:cfex}(b) is consistent from its viewpoint, and here there is a collision on slot 4. 

\begin{figure}[h]
\centerline{
(a)~~
\begin{tabular}{|c|c|c|c|c|} 
\hline 
$i$:  & 1 & 2 & 3 & 4 \\  
\hline
$i.\slotrequest$ & 4 & 3 & 1 & 3\\ 
$i.\msg$ & 1 & 0 & 1 & 0 \\ 
\hline 
$\rr[i]$ & 1 & 0 & 0 & 1\\ 
$\rr[4+i]$ & 1 & 0 & 0 & 1\\ 
\hline
\end{tabular} 
~~~~~(b)~~
\begin{tabular}{|c|c|c|c|c|} 
\hline 
$i$:  & 1 & 2 & 3 & 4 \\  
\hline
$i.\slotrequest$ & 4 & 1 & 1 & 1\\ 
$i.\msg$ & 1 & 1 & 1 & 1 \\ 
\hline 
$\rr[i] $ & 1 & 0 & 0 & 1\\ 
$\rr[4+i]$ & 1 & 0 & 0 & 1\\ 
\hline
\end{tabular} }
\caption{\label{fig:cfex} Collision Freedom is not Common Knowledge} 
\end{figure} 

As mentioned above, we consider in this paper a speculative version of the knowledge-based program, in which an  agent 
transmits its message in its requested slot $s$ in the transmission phase if $\neg K_i \conflict(s)$. 
One could also study a {\em conservative} version, where an agent only transmits if 
$K_i \neg \conflict(s)$. The analysis above shows that this would lead to a much more complicated implementation\footnote{For 
a number of reasons, including the fact that we need an implementation of the knowledge condition at all transmission 
steps, rather than just at the end of the protocol, the above condition is not yet adequate for such an implementation.}, where,
moreover, the agent would transmit only in the low probability case when almost all other agents also have a message to send, and
they happen to pick distinct slots!

Returning to the implementation of the speculative version, 
we need to find the appropriate assignments to the variables $\rcvd0, \rcvd1$ and $\dlvrd$, for which 
we have strong and weak versions. \\

\noindent
{\bf Strong Version:} In this case, reception of a bit $x$ means that the agent knows that some other agent is sending 
that bit $x$. An obvious situation where this is the case is where the agent is not itself sending in the slot, 
the reservation round result is $1$, and the bit $x$ is observed as the round result in the corresponding transmission slot. 
Note that there may still be a collision on that slot, but since the number of agents in the collision is
then odd, at least one must be sending $x$.  As we noted in our previous work \cite{AlBatainehMeyden10}, 
there is another, less obvious, situation when an 
agent can know that another agent is sending a bit $x$ in a slot, viz., when the agent is itself transmitting bit $y$ in that slot and observes that the 
round result for the transmission  is the compliment of $y$. Since the number of other agents in the conflict must be 
even, there must be both another agent sending $0$ and another agent sending $1$ in the slot. 
We have verified by model checking in the case of 3-5 agents that with the assignment
$$
\begin{array}[t]{l}
i.\rcvd x : = \bigvee_{s=1}^n ((i.\slotrequest \neq s \land \rr[s] =1 \land \rr[n+s] =x) ~ \lor \\ 
~~~~~~~~~~~~~~~~~~~~~~~~~(i.\slotrequest =s\land   \rr[s] =1\land  \rr[n+s]\neq i.\msg) )
\end{array} 
 $$
Specification 2 is satisfied in the strong version. 

For the delivery condition, we have verified that the assignment 
$$ 
\begin{array}{l} 
\dlvrd := (\slotrequest \neq 0 \land (C_0 = 0\lor C_0 = 1) ) \lor \\ 
~~~~~         (\slotrequest \neq 0 \land  \msg = 1 ~\land \\
\hspace{1in}   \bigvee_{s\neq t,~s,t =1..n} (\rr[s] = \rr[t] = 1 \land \rr[n+s] = \rr[n+t] = 1)) \lor  \\         
~~~~~          (\slotrequest \neq 0 \land \msg = 0~ \land \\
\hspace{1in}  \bigvee_{s\neq t,~s,t =1..n} (\rr[s] = \rr[t] = 1 \land  \rr[n+s] = \rr[n+t] = 0)) 
\end{array} 
$$ 
works for Specification 3 in the strong version for the cases $n$=3-5. 
The intuitions for this formula are as follows. In the case $C_0 = 0\lor C_0 = 1$, 
as discussed above, it is common knowledge that all slots are conflict-free, so all transmissions are 
guaranteed to be delivered. As just noted, an agent who is not sending on a slot 
receives the value transmitted on that slot. However, an agent sending on a slot, and not noticing 
a clash on the transmission, considers it possible that there are other agents transmitting the
very same value on that slot, and these will not know that there is another agent transmitting on the 
slot. However, if there are at least two distinct reserved slots where that value is transmitted, then 
each receives the value from some slot other than the one  on which it transmits. This is expressed in the remainder of
the formula. \\

\noindent
{\bf Weak Version:} 
In the weak interpretation, we require only that a receiver learn that someone, possibly themselves is sending 
a message. The problem of undetected collsions in the transmission phase does not arise here, 
and the implementation is more straightforward. 
We have verified in the 3-5 agent settings that the following assignments work: 
$$\rcvd x := (\slotrequest \neq 0 \land \msg = x) \lor   \bigvee_{s=1}^n (\rr[s] = 1 \land \rr[n+s] =x)$$ 
$$\dlvrd :=  \slotrequest \neq 0 \land  \bigvee_{s=1}^n (\rr[s] = 1 \land \rr[n+s] = \msg)$$ 
Intuitively, in this case, an agent's own message counts as a delivery, and 
messages observed  on reserved slots can be taken at face value. 

Finally, the anonymity property, Specification 4, has been verified to hold in all the implementations obtained
from the assignments discussed above, when $n = 3-5$. 

\section{Related Work}\label{sec:related} 

Abstractions of the kind we have studied, relating a  protocol  
involving a trusted third party to a
protocol that omits the trusted third party, are often used in  
theoretical studies to specify the
objectives of a multi-party protocol. One example where this is done  
in a formal methods setting
is work by Backes et al \cite{BMM10}, who study the abstraction of pi- 
calculus programs based on
multi-party computations. Where we consider a model checking approach  
to verification, with an expressive epistemic 
specification language, they use a type-checking approach.  Their  
notion  of abstraction is richer than the bisimulation-based approach  
we have taken,
in that they also deal with probabilistic and computational concerns.
However, as we have noted, we are interested in the preservation of a  
set of epistemic
properties (nested knowledge formulas) that is richer in some  
dimensions than is usually considered in this literature.
Our abstraction result could be easily strengthened to incorporate  
probability, as was done for
a secure channel abstraction by van der Meyden and Wilke  
\cite{MeydenWilke07}.
However computational complexity issues mesh less well with
epistemic logic, and developing a satisfactory solution to this  
remains an open problem.

Epistemic model checking is less developed than model checking for  
temporal logic, and
many possible optimization techniques remain to be explored for this  
field. Other approaches
using abstraction in the context of epistemic model checking include
\cite{CohenDLR09,CohenDLQ09}. These works are orthogonal to ours in that
where we are concerned with an abstraction of a particular primitive  
(the Dining Cryptographers protocol), that works for all formulas,  
they are concerned with symmetry reductions or deal with a more  
general class of programs than we have considered,
but need to restrict the class of formulas preserved by the abstraction.

Other model checkers for the logic of knowledge are under development
but MCK remains unique in supporting the perfect recall semantics for  
knowledge
using symbolic techniques.  DEMO \cite{demo} implicitly deals with  
perfect recall,
but is based on a somewhat different logic (epistemic update logic),  
and uses explicit state model checking techniques,
so it is not clear if it could be used for the type of analysis
and scale of programs we have considered in this paper.
MCMAS \cite{mcmas}, MCTK \cite{mctk} and  VERICS \cite{verics} are  
based on the observational semantics
for knowledge (which is also supported in MCK).

It is possible in some cases to represent the perfect
recall semantics using the observational semantics (essentially by  
encoding all history variables into the
state) and this approach is used in \cite{LSGWY} to analyse the same 2- 
phase protocol as
we considered in this paper.
However, this modelling is ad-hoc and the transformation from perfect  
 recall to observational semantics is
handled manually, making it susceptible to missing timing channels if  
not done correctly.
(Moreover, we did briefly experiment with such a modeling for the  
large programs studied in this paper, but
found that the perfect recall model checking algorithms outperform the  
observational semantics
model checking algorithm on these programs.)
The work of \cite{LSGWY} does not view the protocol as a knowledge- 
based program, as we have done,
nor do they consider abstraction.

Knowledge-based programs have been applied successfully in a number of  
applications such as
distributed systems, AI, and game theory. They have been used in  
papers such as
\cite{DM86,Had87,HZ92,BaukusM04,NT93} in order to help in the  
design of new protocols or to clarify
the understanding of existing protocols.
Examples of the development of standard programs from knowledge-based  
programs
can be found in \cite{APPG88,DM86,SR86}.
The approach described in these papers is different from the one we  
discussed here in that it is done by pencil and paper analysis and proof.
Examples of the use of epistemic model checkers to identify  
implementations of knowledge-based programs
remain limited. One is the  work of Baukus and van der Meyden  
\cite{BaukusM04} who use MCK
to analyze several protocols for the cache coherence problem using  
knowledge-based framework.

The 2-phase protocol has been implemented in the Herbivore system  
\cite{Herbivore}, which elaborates it
with protocols allowing agents to enter and exit the system, as well  
as grouping agents in
anonymity cliques for purposes of effciency. Variants of the protocol  
have also been considered by
Pfitzman and Waidner \cite{WP}.
These would make interesting case studies for future applications of  
our approach.

\section{Conclusion}\label{sec:concl} 

We have established the soundness of an abstraction for 
of protocols based on the Dining Cryptographers, and applied this result to 
optimize epistemic model checking of protocols that use Dining Cryptographers
as a primitive. Our experimental results clearly demonstrate that the abstraction 
yields efficiency gains for epistemic model checking in interesting examples. In particular, 
we have used these gains to extend an analysis of a knowledge-based program for the 
2-phase protocol, and derived some interesting conclusions about the subtle information 
flows in the protocol. Several research directions suggest themselves as a result of this work. 
One is to complete the analysis of the knowledge-based program for all numbers of agents. 
We conjecture that our present implementation can be shown to work for all numbers of agents, 
and it would be interesting to have a proof of this claim: this would have to be done manually 
rather than by model checking, unless an induction result can be found for the model checking 
approach.  Another direction is to consider richer extensions of the 2-phase protocol, 
addressing issues such as messages longer than a single bit, 
agent entry and exit protocols, as well as adversarial concerns such as 
collusion, cheating and disruption of the protocol. We hope to address these in future work.

{\bf Acknowledgments:} Thanks to Xiaowei Huang and Kai Englehardt for comments on an earlier version of the paper. 

\bibliographystyle{plain}
\bibliography{references}

\end{document}